\documentclass[12pt]{article}
\usepackage{ifthen}

\makeatletter
\def\@seccntformat#1{\csname the#1\endcsname.\ } 
\def\@biblabel#1{#1.} 
\makeatother

 \newcommand{\qed}{\hfill$\triangle$}

\newif\ifNoRemark
\def\addtheorem#1#2#3#4{
\ifthenelse{\equal{#2}{}}{}%
{\ifthenelse{\expandafter\isundefined\csname the#2\endcsname}{\newcounter{#2}}{}}
\newenvironment{#1}[1][\global\NoRemarktrue]
{\par\addvspace{2mm plus 0.5mm minus 0.2mm}\noindent 
{\bf #3}\ifthenelse{\equal{#2}{}}{}%
{\refstepcounter{#2}{\bf ~\csname the#2\endcsname}}%
{\bf \vphantom{##1}\ifNoRemark.\ \else\ (##1).\fi}\begingroup #4}%
{\endgroup\par\addvspace{1mm plus 0.5mm minus 0.2mm}\global\NoRemarkfalse}
\expandafter\newcommand\csname b#1\endcsname{\begin{#1}}
\expandafter\newcommand\csname e#1\endcsname{\end{#1}}
}

\addtheorem{theorem}{thrm}{Theorem}{\sl}
\addtheorem{lemma}{lmm}{Lemma}{\sl}
\addtheorem{note}{}{Remark}{\rm}
\begin{document}

\title{On the structure of non-full-rank perfect codes}
\author{Olof Heden and Denis S. Krotov\thanks{This research collaboration was partially supported by a grant from Swedish Institute;
the work of the second author
was partially supported by the Federal Target Program ``Scientific and Educational Personnel of Innovation Russia'' for 2009-2013 (government contract No. 02.740.11.0429) and the Russian Foundation for Basic Research (grant 08-01-00673).}}
\maketitle
\begin{abstract} The Krotov combining construction of perfect $1$-error-correcting binary codes from 2000 and a theorem of Heden saying that every non-full-rank perfect $1$-error-correcting binary code can be constructed by this combining construction is generalized to the $q$-ary case. Simply, every non-full-rank perfect code $C$ is
the
union of a well-defined family of $\bar\mu$-components $K_{\bar\mu}$, where $\bar\mu$ belongs to an ``outer'' perfect code $C^\star$, and these components are at distance three from each other. Components from distinct codes can thus freely be combined to obtain new perfect codes. The Phelps general product construction of perfect binary code from 1984 is generalized to obtain $\bar\mu$-components, and new lower bounds on the number of perfect $1$-error-correcting $q$-ary codes are presented.
\end{abstract}
\section{Introduction}
Let $F_q$ denote the finite field with $q$ elements. A \emph{perfect $1$-error-correcting $q$-ary code} of \emph{length} $n$, for short here a \emph{perfect code}, is a subset $C$ of the direct product $F_q^n$, of $n$ copies of $F_q$, having the property that any element of $F_q^n$ differs in at most one coordinate position from a unique element of $C$.

The family of all perfect codes is far from classified or enumerated.
We will in this short note say something about the structure of these codes.
We need the concept of rank.

We consider $F_q^n$ as a vector space of dimension $n$ over the finite field $F_q$. The \emph{rank} of a $q$-ary code $C$, here denoted $\mathrm{rank}(C)$, is the dimension of the linear span $<C>$ of the elements of $C$. Trivial, and well known, counting arguments give that if there exists a perfect code in $F_q^n$ then $n=(q^m-1)/(q-1)$, for some integer $m$, and $|C|=q^{n-m}$. So, for every perfect code $C$,
$$
n-m\leq \mathrm{rank}(C)\leq n\;.
$$
If $\mathrm{rank}(C)=n$ we will say that $C$ has \emph{full rank}.

We will show that every non-full-rank perfect code is a union of so called
\emph{$\bar\mu$-components} $K_{\bar\mu}$, and that these components may
be enumerated by some other perfect code $C^\star$, i.e, $\bar\mu\in C^\star$.
Further, the distance between any two such components will be at least three.
This implies that we will be completely free to combine $\bar\mu$-components
from different perfect codes of same length, to obtain other perfect codes.
Generalizing a construction by Phelps of perfect 1-error correcting binary codes \cite{Phelps84},
we will obtain further $\bar\mu$-components.
As an application of our results we will be able to slightly improve
the lower bound on the number of perfect codes given in \cite{Los06}.

Our results generalize corresponding results for the binary case.
In \cite{Kro:2000:Comb} it was shown
that a binary perfect code can be constructed as the union of different subcodes
($\bar\mu$-components) satisfying some generalized parity-check property,
each of them being constructed independently or taken from another perfect code.
In \cite{Heden2002} it was shown
that every non-full-rank perfect binary code can be obtained by this
combining construction.

\section{Every non-full-rank perfect code is the union of $\bar\mu$-components}
We start with some notation.
Assume we have positive integers $n$, $t$, $n_1$, \ldots, $n_t$ such that $n_1+\dots+n_t\leq n$.
Any $q$-ary word $\bar x$ will be represented
in the block form
$\bar x=(\bar x_1\mid \bar x_2 \mid\ldots\mid\bar x_t\mid\bar x_0)=(\bar x_*\mid\bar x_0)$,
where $\bar x_i=(x_{i1},x_{i2},\ldots,x_{in_i})$, $i=0,1,\ldots,t$,
$n_0=n-n_1-\ldots-n_t$, $\bar x_*=(\bar x_1\mid \bar x_2 \mid\ldots\mid\bar x_t)$.
For every block
$\bar x_i$, $i=1,2,\dots,t$, we define $\sigma_i(\bar x_i)$ by
\[
\sigma_i(\bar x_i)=\sum_{j=1}^{n_i}\,x_{ij}\;,
\]
and, for $\bar x$,
$$
\bar\sigma(\bar x)=\bar\sigma(\bar x_*)=(\sigma_1(\bar x_1),\sigma_2(\bar x_2),\dots,\sigma_t(\bar x_t))
$$

Recall that the Hamming \emph{distance} $\mathrm{d}(\bar x,\bar y)$ between two words $\bar x$, $\bar y$ of the same length means the number of positions in which they differ.

A {\it monomial transformation} is a map of the space $F_q^n$ that
can be composed by a permutation of the set of coordinate positions
and the multiplication in each coordinate position with some non-zero
element of the finite field $F_q$.

A $q$-ary code $C$ is \emph{linear} if $C$ is a subspace of $F_q^n$. A linear perfect code is called a \emph{Hamming code}.

\begin{theorem}
Let $C$ be any non-full-rank perfect code $C$ of length $n=(q^m-1)/(q-1)$. To any integer $r<m$, satisfying
$$
1\leq r \leq n-\mathrm{rank}(C)\;,
$$
there is a $q$-ary Hamming code $C^\star$ of length $t=(q^{r}-1)/(q-1)$, such that for  some monomial transformation $\psi$
$$
\psi(C)=\bigcup_{\bar \mu\in C^\star}\, K_{\bar \mu}\;,
$$
where
\begin{equation}\label{eq:def-comp}
K_{\bar\mu}=\{(\bar x_1\mid\bar x_2\mid\dots\mid\bar x_t\mid \bar x_{0})\,:\,\bar\sigma(\bar x)=\bar\mu,\;\;\bar x_1,\bar x_2,\ldots,\bar x_t\in F_q^{q^{s}},\;\;\bar x_{0}\in C_{\bar \mu}(\bar x_*)\,\}\;
\end{equation}
for some family of perfect codes $C_{\bar\mu}(\bar x)$, of length $1+q+q^2+\dots+q^{s-1}$, where $s=m-r$, and satisfying, for each $\bar\mu\in C^\star$,
\begin{equation}\label{eq:empty-intersec}
\mathrm{d}(\bar x_*, \bar x'_*)\leq2\qquad\Longrightarrow\qquad C_{\bar\mu}(\bar x_*)\cap C_{\bar\mu}(\bar x'_*)=\emptyset\;.
\end{equation}
\end{theorem}

The code $C^\star$ will be called an \emph{outer} code to $\psi(C)$. The subcodes $K_{\bar\mu}$ will be called \emph{$\bar\mu$-components} of $\psi(C)$. As the minimum distance of $C$ is three, the distance between any two distinct $\bar\mu$-components will be at least three.

\begin{proof}
Let $D$ be any subspace of $F_q^n$ containing $<C>$, and of dimension $n-r$.
By using a monomial transformation $\psi$ of space we may achieve that the dual space of $\psi(D)$
is the nullspace of a $r\times n$-matrix
\[
H=\left[\begin{array}{ccc|ccc|c|ccc|ccc}
|&&|&|&&|&&|&&|&|&&|\\
\bar\alpha_{11}&\cdots&\bar\alpha_{1n_1}&\bar\alpha_{21}&\cdots&\bar\alpha_{2n_2}&\cdots&\bar\alpha_{t1}&\cdots&\bar\alpha_{tn_t}&\bar 0&\cdots&\bar0\\
|&&|&|&&|&&|&&|&|&&|
\end{array}\right]
\]
where $\bar\alpha_{ij}=\bar\alpha_i$, for $i=1,2,\dots,t$, the first non-zero coordinate in each vector $\bar\alpha_i$ equals 1, $\bar \alpha_i\neq\bar\alpha_{i'}$, for $i\neq i'$, and where the columns of $H$ are in lexicographic order, according to some given ordering of $F_q$.

To avoid too much notation we assume that $C$ was such that $\psi=\mathrm{id.}$

Let $C^\star$ be the null space of the matrix
\[
H^\star=\left[\begin{array}{cccc}
|&|&&|\\
\bar\alpha_1&\bar\alpha_2&\cdots&\bar\alpha_t\\
|&|&&|
\end{array}\right]
\]
Define, for $\bar \mu\in C^\star$,
\[
K_{\bar \mu}=\{\,(\bar x_1\mid\bar x_2\mid\dots\mid \bar x_t\mid \bar x_0)\in C\,:\,(\sigma_1(\bar x_1),\sigma_2(\bar x_2),\dots,\sigma(\bar x_t))=\bar \mu\,\}\;.
\]
Then,
\[
C=\bigcup_{\bar \mu\in C^\star}\,K_{\bar \mu}.
\]
Further, since any two columns of $H^\star$ are linearly independent,
for any two distinct words $\bar\mu$ and $\bar\mu'$ of $C^\star$
\begin{equation}\label{eq:min-dist-komp}
\mathrm{d}(K_{\bar \mu}, K_{\bar \mu'})\geq 3.
\end{equation}
We will show that $K_{\bar\mu}$ has the properties given in Equation (\ref{eq:def-comp}).

Any word $\bar x=(\bar x_1\mid\bar x_2\mid\dots\mid\bar x_t\mid\bar x_{0})$ must be at distance at most one from a word of $C$, and hence, the word $(\sigma_1(\bar x_1),\sigma_2(\bar x_2),\dots,\sigma_t(\bar x_t))$ is at distance at most one from some word of $C^\star$. It follows that $C^\star$ is a perfect code, and as a consequence, as $C^\star$ is linear, it is a Hamming code with parity-check matrix $H^\star$.
As the number of rows of $H^\star$ is $r$, we then get that the number $t$ of columns of $H^\star$ is equal to
\[
t=\frac{q^{r}-1}{q-1}=1+q+q^2+\dots+q^{r-1}\;.
\]

For any word
$\bar x_*$
of $F_q^{n_1+n_2+\dots+n_t}$ with $\bar\sigma(\bar x_*)=\bar\mu\in C^\star$, we now define the code $C_{\bar\mu}(\bar x_*)$ of length $n_0$ by
\[
C_{\bar\mu}(\bar x_*)=\{\,\bar c\in F_q^{n_0}:
(\bar x_*\mid\bar c)
\in C\,\}\;.
\]
Again, using the fact that $C$ is a perfect code, we may deduce that for any $\bar x_*$ such that the set $C_{\bar\mu}(\bar x_*)$ is non empty, the set $C_{\bar\mu}(\bar x_*)$ must be a perfect code of length $n_0=(q^{s}-1)/(q-1)$, for some integer $s$.

From the fact that the minimum distance of $C$ equals three, we get the property in Equation (\ref{eq:empty-intersec}).

Let $\bar e_i$ denote a word of weight one with the entry 1 in the coordinate position $i$.
It then follows that the two perfect codes
$C_{\bar\mu}(\bar x_*)$ and $C_{\bar\mu}(\bar x_*+\bar e_1-\bar e_i)$, for $i=2,3,\dots,n_1$, must be mutually disjoint. Hence, $n_1$ is at most equal to the number of perfect codes in a partition of $F_q^{n_0}$ into perfect codes, i.e.,
\[
n_1\leq (q-1)n_0+1=q^{s}\;.
\]
Similarly,  $n_i\leq q^{s}$, for  $i=2,3,\dots,t$.

Reversing these arguments, using Equation (\ref{eq:min-dist-komp}) and the fact that $C$ is a perfect code, we find that $n_i$, for each $i=1,2,\dots,t$, is at least equal to the number of words in an $1$-ball of $F_q^{n_0}$.

We conclude that $n_i=q^{s}$, for $i=1,2,\dots,t$, and finally
\[
n= q^{s}(1+q+q^2+\dots+q^{r-1})+1+q+q^2+\dots+q^{s-1}=1+q+q^2+\dots+q^{r+s-1}\;.
\]
Given $r$, we can then find $s$ from the equality
$$
n=1+q+q^2+\dots+q^{m-1}\;.
$$
\qed\end{proof}

\section{Combining construction of perfect codes}

In the previous section, it
was
shown that a perfect
code, depending on its rank, can be divided onto small or large
number of so-called $\bar\mu$-components, which satisfy some equation
with $\bar\sigma$.
The construction described in the following theorem
realizes the idea of combining independent $\bar\mu$-components,
differently constructed or taken from different perfect codes, in one perfect code.

A function $f:\Sigma^n\to\Sigma$,
where $\Sigma$ is some set, is called an \emph{$n$-ary} (or \emph{multary}) \emph{quasigroup} of order $|\Sigma|$ if in the equality $z_0=f(z_1,\ldots,z_n)$ knowledge of any $n$ elements
of $z_0$, $z_1$, \ldots, $z_n$ uniquely specifies the remaining one.
\begin{theorem}\label{th:constr}
Let $m$ and $r$ be integers, $m>r$, $q$ be a prime power, $n=(q^m-1)/(q-1)$ and $t=(q^r-1)/(q-1)$. Assume that $C^*$ is
a perfect code in $F_q^t$ and for every
$\bar\mu\in C^*$ we have
a distance-$3$ code $K_{\bar\mu}\subset F_q^n$ of cardinality $q^{n-m-(t-r)}$
 that satisfies the following generalized parity-check
law:
$$ \bar\sigma(\bar x)=(\sigma_1(x_{1},\ldots,x_{l}),\ldots,\sigma_t(x_{lt-l+1},\ldots,x_{lt}))=\bar\mu $$
for every $\bar x=(x_1,\ldots,x_n)\in K_{\bar\mu}$, where
$l=q^{m-r}$ and
$\bar\sigma=(\sigma_1,\ldots,\sigma_t)$ is a collections of $l$-ary
quasigroups of order $q$. Then the union
$$C=\bigcup_{\bar\mu\in C^*} K_{\bar\mu}$$ is a perfect code in $F_q^n$.
\end{theorem}
\begin{proof}
It is easy to check that $C$ has the cardinality of a perfect code.
The distance at least $3$ between different words $\bar x$, $\bar y $ from $C$ follows from the code distances of $K_{\bar\mu}$
(if $\bar x$, $\bar y $ belong to the same $K_{\bar\mu}$) and $C^*$
(if $\bar x$, $\bar y $ belong to different $K_{\bar\mu'}$, $K_{\bar\mu''}$, $\bar\mu',\bar\mu''\in C^*$).
\qed\end{proof}

The \emph{$\bar\mu$-components} $K_{\bar\mu}$ can be constructed independently or taken from
different perfect codes. In the important case when all $\sigma_i$ are linear quasigroups (e.g., $\sigma_i(y_1,\ldots,y_l)=y_1+\dots+y_l$)
the components can be taken from any perfect code of rank at most $n-r$, as follows from the previous section (it should be noted that if $\bar\sigma$ is linear, then a $\bar\mu$-component can be obtained from any $\bar\mu'$-component by adding a vector $\bar z$ such that $\bar\sigma(\bar z)=\bar\mu-\bar\mu'$).

In general, the existence of $\bar\mu$-components that satisfy
the generalized parity-check
law
 for arbitrary $\bar \sigma$
is questionable.
But for some class of $\bar \sigma$ such components exist, as we will see from the following two
subsections.

\begin{note}
It is worth mentioning that $\bar\mu$-components can exist for arbitrary length $t$ of $\bar\mu$
(for example, in the next two subsections there are no restrictions on $t$),
if we do not require the possibility to combine them into a perfect code.
This is especially important for the study of perfect codes of small ranks (close to the rank of a linear perfect code): once we realize that the code is the union of $\bar\mu$-components of some special form,
we may forget about the code length and consider $\bar\mu$-components for arbitrary length of $\bar\mu$,
which allows to use recursive approaches.
\end{note}

\subsection{Mollard-Phelps construction}
Here we describe the way to construct
$\bar\mu$-components
derived from
the product construction discovered independently in \cite{Mollard}
and \cite{Phelps:q}.
In terms of $\bar\mu$-components,
the construction in \cite{Phelps:q} is more general; it allows substitution of arbitrary multary quasigroups, and we will use this possibility in Section~\ref{s:No}.

\begin{lemma}\label{l:Ph}
Let $\bar\mu\in F_q^{t}$ and
let $C^\#$ be a perfect code in $F_q^{k}$.
Let $v$ and $h$ be
$(q-1)$-ary quasigroups of order $q$ such that
the code $\{(\bar y\mid v(\bar y)\mid h(\bar y))\,:\, \bar y\in F_q^{q-1}\}$ is perfect.
Let $V_1$, \ldots, $V_t$ and $H_1$, \ldots, $H_k$ be respectively $(k+1)$-ary and $(t+1)$-ary quasigroups of order $q$.
Then the set
\begin{eqnarray*}
 K_{\bar\mu} &=& \Big\{ \big(
\underline{\bar x_{11}\mid ...\mid \bar x_{1k}\mid  y_1} \mid
\underline{\bar x_{21}\mid ...\mid \bar x_{2k}\mid  y_2} \mid
\ldots\mid
\underline{\bar x_{t1}\mid ...\mid \bar x_{tk}\mid  y_t}\mid
\underline{z_1\mid z_2\mid  ...\mid  z_k}\big)
: \\ &&\ \
 \bar x_{ij}\in F_q^{q-1},
\\ &&\ \  \left(V_1(v(\bar x_{11}),...,v(\bar x_{1k}),y_1),\ldots,V_t(v(\bar x_{t1}),...,v(\bar x_{tk}),y_t)\right)=\bar\mu, \\ &&\ \
\left(H_1(h(\bar x_{11}),...,h(\bar x_{t1}),z_1),\ldots,H_k(h(\bar x_{1k}),...,h(\bar x_{tk}),z_k)\right)\in C^{\#} \Big\}
\end{eqnarray*}
is a $\bar\mu$-component that satisfies the generalized parity-check
law
with $$\sigma_i(\cdot,\ldots,\cdot,\cdot)=V_i(v(\cdot),...,v(\cdot),\cdot).$$
(The elements of $F_q^{(q-1)kt+k+t}$ in this construction may be thought of as
three-dimensional arrays where the elements of $\bar x_{ij}$ are z-lined, every underlined block is y-lined,
and the tuple of blocks is x-lined. Naturally, the multary quasigroups $V_i$ may be named ``vertical'' and $H_i$, ``horizontal''.)
\end{lemma}
The proof of the code distance is similar to that in \cite{Phelps:q}, and the other properties of a $\bar\mu$-component are straightforward.
The existence of admissible $(q-1)$-ary quasigroups $v$ and $h$ is the only restriction on the $q$ (this concerns the next subsection as well).
If $F_q$ is a finite field, there are linear examples:
$v(y_1,\ldots,y_{q-1})=y_1+\ldots+y_{q-1}$, $v(y_1,\ldots,y_{q-1})=\alpha_1 y_1+\ldots+\alpha_{q-1} y_{q-1}$
where $\alpha_1$, \ldots, $\alpha_{q-1}$ are all the non-zero elements of $F_q$.
If $q$ is not a prime power, the existence of a $q$-ary perfect code
of length $q+1$ is an open problem
(with the only exception $q=6$, when the nonexistence follows from the nonexistence of two orthogonal $6\times 6$ Latin squares \cite[Th.\,6]{GolombPosner64}).

\subsection{Generalized Phelps construction}
Here we describe another way to construct
$\bar\mu$-components,
which generalizes the construction of binary perfect codes
from \cite{Phelps84}.

\begin{lemma}\label{l:Ph2}
Let $\bar\mu\in F_q^{t}$.
Let for every $i$ from $1$ to $t+1$ the codes $C_{i,j}$, $j=0,1,\ldots,qk-k$ form a partition
of $F_q^{k}$ into perfect codes and
$\gamma_i:F_q^{k}\to \{0,1,\ldots,qk{-}k\}$ be the corresponding partition function:
$$ \gamma_i(\bar y)=j\ \Longleftrightarrow\ \bar y\in C_{i,j}. $$
Let $v$ and $h$ be
$(q-1)$-ary quasigroups of order $q$ such that
the code $\{(\bar y\mid v(\bar y)\mid h(\bar y))\,:\, \bar y\in F_q^{q-1}\}$ is perfect.
Let $V_1$, \ldots, $V_t$ be $(k+1)$-ary quasigroups of order $q$ and
$Q$ be a $t$-ary quasigroup of order $qk-k+1$.

\begin{eqnarray*}
 K_{\bar\mu} &=& \Big\{ \big(
\underline{\bar x_{11}\mid ...\mid \bar x_{1k}\mid  y_1} \mid
\underline{\bar x_{21}\mid ...\mid \bar x_{2k}\mid  y_2} \mid
\ldots\mid
\underline{\bar x_{t1}\mid ...\mid \bar x_{tk}\mid  y_t}\mid
\underline{z_1\mid z_2\mid  ...\mid  z_k}\big)
: \\ &&\ \
 \bar x_{ij}\in F_q^{q-1},
\\ &&\ \  \left(V_1(v(\bar x_{11}),...,v(\bar x_{1k}),y_1),\ldots,V_t(v(\bar x_{t1}),...,v(\bar x_{tk}),y_t)\right)=\bar\mu, \\ &&\ \ Q(\gamma_1(h(\bar x_{11}),...,h(\bar x_{1k})),\ldots,\gamma_t(h(\bar x_{t1}),...,h(\bar x_{tk})))=\gamma_{t+1}(z_1,...,z_k) \Big\}
\end{eqnarray*}
is a $\bar\mu$-component that satisfies the generalized parity-check
law
with $$\sigma_i(\cdot,\ldots,\cdot,\cdot)=V_i(v(\cdot),...,v(\cdot),\cdot).$$
\end{lemma}
The proof consists of trivial verifications.

\section{On the number of perfect codes}\label{s:No}

In this section we discuss some observations,
which result in the best known lower bound on the
number of $q$-ary perfect codes, $q\geq 3$.
The basic facts are already contained in other known results:
lower bounds on the number of multary quasigroups of order $q$, the construction
\cite{Phelps:q} of perfect codes from multary quasigroups of order
$q$, and the possibility to choose the quasigroup independently for
every vector of the outer code (this possibility was not explicitly
mentioned in \cite{Phelps:q}, but used in the previous paper
\cite{Phelps84}).

A general lower bound, in terms of the number of multary quasigroups,
is given by Lemma~\ref{l:qua-code}. In combination with Lemma~\ref{l:Nqua}, it gives
explicit numbers.
\begin{lemma}\label{l:qua-code}
The number of $q$-ary perfect codes of length $n$ is not less than
$$Q\left(\frac{n-1}q,q\right)^{R_{\frac{n-1}q}}$$ where $Q(m,q)$ is the number of
$m$-ary quasigroups of order $q$ and
where $R_{n'}=q^{n'}/(n'q-q+1)$ is the
cardinality of a perfect code of length $n'$.
\end{lemma}
\begin{proof}
Constructing a perfect code like in Theorem~\ref{th:constr} with
$t=\frac{n-1}q$, we combine $R_{\frac{n-1}q}$ different
$\bar\mu$-components.

 Constructing every such a component as in Lemma~\ref{l:Ph2}, $k=1$, $t=\frac{n-1}q$,
we are free to choose the $t$-ary quasigroup $Q$ of order $q$ in $Q(t,q)$ ways.
Clearly, different $t$-ary quasigroups give different components.
(Equivalently, we can use Lemma~\ref{l:Ph} and choose the $(t+1)$-ary quasigroup $H_1$,
but should note that the value of $H_1$ in the construction is always fixed when $k=1$, because $C^\#$ consists of only one vertex;
so we again have $Q(t,q)$ different choices, not $Q(t+1,q)$).
\qed\end{proof}

\begin{lemma}\label{l:Nqua}The number $Q(m,q)$ of $m$-ary quasigroups of order $q$
satisfies:
\begin{enumerate}
  \item[\rm(a)] {\rm\cite{LaywineMullen}} $Q(m,3)=3\cdot 2^m;$
  \item[\rm(b)] {\rm\cite{PotKro:asymp}} $Q(m,4) = 3^{m+1}\cdot 2^{2^m+1}(1+o(1));$
  \item[\rm(c)] {\rm\cite{KPS:ir}} $Q(m,5)\geq 2^{3^{n/3-0.072}};$
  \item[\rm(d)] {\rm\cite{PotKro:numberQua}} $Q(m,q)\geq
  2^{((q^2-4q+3)/4)^{n/2}}$ for odd $q$ {\rm(}the previous bound {\rm\cite{KPS:ir}} was
$Q(m,q)\geq 2^{\lfloor q/3\rfloor^{n}}${\rm);}
  \item[\rm(e)] {\rm\cite{KPS:ir}} $Q(m,q_1 q_2)\geq Q(m,q_1)\cdot Q(m,q_2)^{q_1^m}$.
\end{enumerate}
\end{lemma}

For odd $q\geq 5$, the number of codes given by
Lemmas~\ref{l:qua-code} and~\ref{l:Nqua}(c,d) improves the constant
$c$ in the lower estimation of form $e^{e^{cn(1+o(1))}}$ for the
number of perfect codes, in comparison with the last known lower
bound \cite{Los06}. Informally, this can be explained in the
following way: the construction in \cite{Los06} can be described in
terms of mutually independent small modifications of the linear
multary quasigroup of order $q$, while the lower bounds in
Lemma~\ref{l:Nqua}(c,d) are based on a specially-constructed
nonlinear multary quasigroup that allows a lager number of
independent modifications. For $q= 3$ and $q=2^s$, the number of
codes given by Lemmas~\ref{l:qua-code} and~\ref{l:Nqua}(a,b,e) also
slightly improves the bound in \cite{Los06}, but do not affect on
the constant $c$.


\providecommand\href[2]{#2} \providecommand\url[1]{\href{#1}{#1}} \def\DOI#1{
  {DOI}: \href{http://dx.doi.org/#1}{#1}}\def\DOIURL#1#2{ {DOI}:
  \href{http://dx.doi.org/#2}{#1}}\providecommand\bbljun{June}

\bigskip
\noindent
O. Heden\\
Department of Mathematics,
KTH\\
S-100 44 Stockholm,
Sweden\\
email: \texttt{olohed@math.kth.se}

\bigskip
\noindent
D. Krotov\\
Sobolev Institute of Mathematics\\
and\\
Mechanics and Mathematics Department,
Novosibirsk State University\\
Novosibirsk, Russia\\
email: \texttt{krotov@math.nsc.ru}

\end{document}